\documentclass[11pt]{article}
\usepackage{amsfonts}
\usepackage{amssymb}
\usepackage{graphicx}
\usepackage{amsmath}
\usepackage{makeidx}
\usepackage[T1]{fontenc}
\usepackage{color}

\newtheorem{theorem}{Theorem}

\newtheorem{adefinition}[theorem]{Definition}

\newtheorem{aexample}[theorem]{Example}

\newtheorem{lemma}[theorem]{Lemma}

\newtheorem{aproblem}[theorem]{Problem}
\newtheorem{proposition}[theorem]{Proposition}
\newtheorem{acomment}[theorem]{Comment}

\newtheorem{aremark}[theorem]{Remark}
\newenvironment{remark}{\begin{aremark}\rm}{\end{aremark}}

\newenvironment{proof}[1][Proof]{\textbf{#1.} }{\ \rule{0.5em}{0.5em}}
\numberwithin{equation}{section} \numberwithin{theorem}{section}

\DeclareMathOperator{\tr}{Tr}

\newcommand{\trip}{\|\kern-.08em |}
\newcommand\beq{\begin{equation}}
\newcommand\eeq{\end{equation}}

\makeatother
\makeindex
\begin{document}

\title{Szeg\"o-type Theorems for \\
One-Dimensional Schrodinger Operator \\
with Random Potential\\
(smooth case)}
\author{L. Pastur, M. Shcherbina \\
Mathematical Division, B.Verkin Institute \\
for Low Temperature Physics and Engineering \\
Kharkiv, Ukraine}
\maketitle
\date{}

\begin{abstract}
The paper is a continuation of work \cite{Ki-Pa:15} in which the general
setting for analogs of the Szeg\"o theorem for ergodic operators was given
and several interesting cases were considered. Here we extend the results of
\cite{Ki-Pa:15} to a wider class of test functions and symbols which
determine the Szeg\"o-type asymptotic formula for the one-dimensional
Schrodinger operator with random potential. We show that in this case the
subleading term of the formula is given by a Central Limit Theorem in the
spectral context, hence the term is asymptotically proportional to $L^{1/2}$%
, where $L$ is the length of the interval on which the Schrodinger operator
is initially defined. This has to be compared with the classical Szeg\"o
formula, where the subleading term is bounded in $L, \; L \to \infty$. We
prove an analog of standard Central Limit Theorem (the convergence of the
probability of the corresponding event to the Gaussian Law) as well as an
analog of the almost sure Central Limit Theorem (the convergence with
probability 1 of the logarithmic means of the indicator of the corresponding
event to the Gaussian Law). We illustrate our general results by
establishing the asymptotic formula for the entanglement entropy of free
disordered Fermions for non-zero temperature.
\end{abstract}


\section{Introduction}

The Szeg\"{o} theorem (also known as the strong Szeg\"{o} theorem) is an
interesting asymptotic formula for the restrictions of functions of the
Toeplitz operators as the size of the domain of restriction tends to
infinity. It has a number of applications and extensions pertinent to
analysis, mathematical physics, operator theory, probability theory and
statistics and (recently) quantum information theory, see \cite%
{Bi:12,Bo-Si:90,De-Co:13,Si:12,So:13}. In this paper we consider an
extension of the theorem viewed as an asymptotic trace formula for a certain
class of selfadjoint operators. We will start with an outline of the
continuous version of the Szeg\"{o} theorem presenting it in the form which
explains our motivation.

Let  $k:\mathbb{R}\rightarrow \mathbb{R}$ be an even and sufficiently smooth function
from $L^{1}(\mathbb{R})$,%
\begin{equation}
\Lambda =[-M,M],\;|\Lambda |=2M,  \label{la}
\end{equation}%
$K$ and $K_{\Lambda }:=K|_{\Lambda }$ be selfadjoint convolution operators in $%
L^{2}(\mathbb{R})$ and its restriction to $L^{2}(\Lambda )$ given by
\begin{eqnarray}\label{coco}
(Ku)(x) &=&\int_{-\infty }^{\infty }k(x-y)u(y)dy,\;x\in \mathbb{R}, \\
(K_{\Lambda }u)(x) \notag &=&\int_{-M}^{M}k(x-y)u(y)dy,\;x\in \Lambda
\end{eqnarray}%
Set $A=\mathbf{1}_{L^{2}(\mathbb{R})}+K$ and $A_{\Lambda }= \mathbf{1}%
_{L^{2}(\Lambda)}+K_{\Lambda}$ and consider $\varphi :\mathbb{R}\rightarrow
\mathbb{R}$ such that $\varphi (A_{\Lambda })$ is of trace class in $%
L^{2}(\Lambda )$. Then we have according to Szeg\"o and subsequent works%
\begin{equation}
\tr_\Lambda \varphi (A_{\Lambda })=|\Lambda |\int_{-\infty }^{\infty
}\varphi (a(t))dt+\mathcal{T}+o(1),\;|\Lambda |\rightarrow \infty ,
\label{szcl}
\end{equation}%
where $\tr_{\Lambda }$ is the trace in $L^2(\Lambda)$, $a(t)=1+\widehat{k}(t),\;t\in
\mathbb{R}$, $\widehat{k}$ is the Fourier transform of $k$ and the subleading term $\mathcal{T}$ is a $\Lambda$-independent
functional of $\varphi $ and $a$. We will call $\varphi $ and $a$ the test
function and the symbol respectively. 

Let $P=i\frac{d}{dx}$ be the selfadjoint operator in $L^{2}(\mathbb{R})$.
Then the r.h.s. of is $\tr_{\Lambda }\varphi (a_{\Lambda }(P))$, i.e., is
determined by the triple $(\varphi ,a,P)$, and since $a$ is even and smooth enough, we have $a(x)=b(x^2)$, hence the triple $(\varphi,b,P^2)$. It was proposed in \cite{Ki-Pa:15}
to consider instead of $P^2$ the Schrodinger operator $H=P^{2}+V$ where the
potential $V:\mathbb{R}\rightarrow \mathbb{R}$ is an ergodic process. It
seems that the replacement is of interest in itself since the ergodicity of
the potential guarantees the sufficient regular large $\Lambda $ behavior of
$\tr_{\Lambda }\varphi (a_{\Lambda }(H))$, hence a well defined asymptotic
formulas. Besides, the quantity $\tr_{\Lambda }\varphi (a_{\Lambda }(H))$
for certain $\varphi ,a$ and $V$ arises in quantum information theory and
quantum statistical mechanics, see \cite{Ei-Co:11}, Remark \ref{r:renyi} and
references therein.

Similar setting is also possible in the discrete case. In fact, it
is this case of which was initially studied by Szeg\"o for Toeplitz operators, while the continuous
case outlined above was considered later by Akhiezer, Kac and Widom, see e.g. \cite{Bo-Si:90} for a review.
We will also consider in this paper the discrete case.

In \cite{Ki-Pa:15} simple but rather non-trivial discrete cases were
studied. There $a(x)=x$ and $\varphi $ is $\ (x-x_0)^{-1}$ or $\log (x-x_0)$
where $x_0$ is outside the spectrum of the discrete Schrodinger operator
with ergodic potential (random and almost periodic). In particular, it was
shown that if the potential in the discrete Schrodinger equation is a
collection of independent identically distributed (i.i.d.) random variables,
then the leading term on the right of the analog of (\ref{szcl}) is again of
the order \ $|\Lambda |$ and is not random, but the subleading term is of
the order $|\Lambda |^{1/2}$ and is a Gaussian random variable. In fact, a
certain Central Limit Theorem for an appropriately normalized quantity $\tr%
_\Lambda \varphi (a_{\Lambda }(H))$ was established. In this paper we extend
this result for those $\varphi $ and $a$ which, roughly speaking, have the
Lipshitz derivative (see condition (\ref{afcon}) below). Note that similar
conditions were used Szeg\"o in his pioneering works, although the
conditions were seriously weakened in subsequent works, see \cite%
{Bo-Si:90,De-Co:13,Si:12,So:13}.

\section{Problem and Results}

Let $H$ be the one-dimensional Schrodinger operator in $l^{2}(\mathbb{Z}) $%
\begin{equation}
H=H_{0}+V,  \label{h}
\end{equation}%
where%
\begin{equation}
(H_{0}u)_{j}=-u_{j+1}-u_{j-1},\;j\in \mathbb{Z}  \label{h0}
\end{equation}%
and%
\begin{equation}
(Vu)_{ j}=V_{j}u_{j},\;j\in \mathbb{Z},\;|V_{j}|\leq \overline{V}<\infty
\label{qq}
\end{equation}%
is a potential which we assume to be a sequence of independent and
identically distributed (i.i.d.) random variables bounded for the sake of
technical simplicity.

The spectrum $\sigma (H)$ is a non-random closed set and
\begin{equation}
\sigma (H)\subset K:=[-2-\overline{V},2+\overline{V}],  \label{spk}
\end{equation}%
see \cite{Pa-Fi:92}.


Let also $a :\sigma(H)\rightarrow \mathbb{R}$ (symbol) and $%
\varphi:a(\sigma(H))\rightarrow \mathbb{R}$ (test function) be bounded
functions. Introduce the integer valued interval (cf. (\ref{la}))
\begin{equation}
\Lambda =[-M,M]\subset \mathbb{Z},\;|\Lambda |=2M+1  \label{lm}
\end{equation}%
and the operator $\chi _{\Lambda }: l^2(\mathbb{Z}) \to l^2(\Lambda)$ of
restriction, i.e., if $x=\{x_{j}\}_{j\in \mathbb{Z}}\in l^{2}(\mathbb{Z})$,
then $\chi _{\Lambda }x=x_{\Lambda }:=\{x_{j}\}_{j\in \Lambda }\in
l^{2}(\Lambda )$. For any operator $A=\{A_{jk}\}_{j,k\in \mathbb{Z}}$ in $%
l^2(\mathbb{Z})$ we denote (cf. (\ref{coco}))%
\begin{equation}
A_{\Lambda }:=\chi _{\Lambda }A\chi _{\Lambda }=\{A_{jk}\}_{j,k\in \Lambda }
\label{ala}
\end{equation}%
its restriction to $l^{2}(\Lambda )$. Note that the spectra of $A$ and $%
A_{\Lambda }$ are related as follows%
\begin{equation}
\sigma (A_{\Lambda })\subset \sigma (H)  \label{ssl}
\end{equation}

Our goal is to study the asymptotic behavior of
\begin{equation}
\mathrm{Tr}_{\Lambda }\varphi (a_{\Lambda }(H)):=\sum_{j \in \Lambda}
(\varphi (a_{\Lambda }(H)))_{jj},\;|\Lambda |\rightarrow \infty,
\label{trfl}
\end{equation}%
where
\begin{equation}  \label{trl}
\tr_{\Lambda }...=\tr \chi_{\Lambda}...\chi_{\Lambda}
\end{equation}
As was mentioned above, this problem dates back to works of Szeg\"{o} \cite%
{Gr-Sz:58} and has been extensively studied afterwards for the Toeplitz and
convolution operators, see e.g., \cite{Bo-Si:90,De-Co:13,Si:12} and
references therein. Recall that any sequence
\begin{equation}
\{A_{j}\}_{j\in \mathbb{Z}},\;\overline{A_{j}}=A_{-j},\;\sum_{j\in \mathbb{Z}%
}|A_{j}|<\infty  \label{aj}
\end{equation}%
determines a selfadjoint (discrete convolution) operator in $l^{2}(\mathbb{Z})$, cf. (\ref{coco})
\begin{equation}
A=\{A_{j-k}\}_{j,k\in \mathbb{Z}},\;(Au)_{j}=\sum_{k\in \mathbb{Z}%
}A_{j-k}u_{k}.  \label{Ac}
\end{equation}%
Let
\begin{equation*}
a(p)=\sum_{j\in \mathbb{Z}}A_{j}e^{2\pi ipj},\;p\in \mathbb{T}=[0,1)
\end{equation*}%
be the Fourier transform of $\{A_{j}\}_{j\in \mathbb{Z}}$. Then, according
to Szeg\"{o} (see e.g. \cite{Gr-Sz:58}), if $\varphi $ and $a$ are
sufficiently regular, then we have the two-term asymptotic formula (cf \ref%
{szcl})%
\begin{equation}
\tr_\Lambda \varphi (A_{\Lambda })=|\Lambda |\int_{\mathbb{T}}\varphi
(a(t))dt+\mathcal{T}+o(1),\;|\Lambda |\rightarrow \infty ,  \label{sf2}
\end{equation}%
where the subleading term $\mathcal{T}$ is again a $\Lambda $-independent functional of
$\varphi $ and $a$. Note that the traditional setting for the
Szeg\"{o} theorem uses the Toeplitz operators defined by the semi-infinite
matrix $\{A_{j-k}\}_{j,k\in \mathbb{Z}_{+}}$ and acting in $l^{2}(\mathbb{Z}%
_{+})$. The restrictions of Toeplitz operators are the upper left blocks $%
\{A_{j-k}\}_{j,k=0}^{L}$ of $\{A_{j-k}\}_{j,k=0}^{%
\infty }$. On the other hand, we will use in this paper the convolution
operators (\ref{Ac}) defined by the double infinite matrix $%
\{A_{j-k}\}_{j,k\in \mathbb{Z}}$, acting in $l^{2}(\mathbb{Z})$ and having
their central $L \times L, \; L=2M+1$ blocks as restrictions. The latter
setting seems more appropriate for the goal of this paper dealing with
ergodic operators where the setting seems more natural. The same setting is
widely used in multidimensional analogs of Szeg\"{o} theorem \cite{Bo-Si:90}.

Note now that the convolution operators in $l^{2}(\mathbb{Z}^{d})$ and $%
L^{2}(\mathbb{R}^{d}), \; d \ge 1$ admit a generalization, known as ergodic
(or metrically transitive) operators, see \cite{Pa-Fi:92}. We recall their
definition in the (discrete) case of $l^{2}(\mathbb{Z}).$

Let $(\Omega ,\mathcal{F},P)$ be a probability space and $T$ is an ergodic
automorphism of the space. A measurable map $A=\{A_{jk}\}_{j,k\in \mathbb{Z}%
} $ from $\Omega $ to bounded operators in $l^{2}(\mathbb{Z})$ is called
ergodic operator if we have with probability 1 for every $t\in \mathbb{Z}$
\begin{equation}
A_{j+t,k+t}(\omega )=A_{jk}(T^{t}\omega ),\;\forall j,k\in \mathbb{Z}.
\label{eom}
\end{equation}%
Choosing $\Omega =\{0\}$, we obtain from (\ref{eom}) that $A$ is a
convolution operator (\ref{Ac}). Thus, ergodic operators comprise a
generalization of convolution operators, while the latter can be viewed as
non-random ergodic operators.

It is easy to see that the discrete Schrodinger operator with ergodic
potential (\ref{h}) -- (\ref{qq}) is an ergodic operator. Moreover, if $%
\sigma (H)$ is the spectrum of $H$, then $\sigma (H)$ is non-random, for any
bounded and measurable $f:\sigma (H)\rightarrow \mathbb{R}$ the operator $%
f(H)$ is also ergodic and if $\{f_{jk}\}_{j,k\in \mathbb{Z}}$ is its matrix,
then $\{f_{jj}\}_{j\in \mathbb{Z}}$ is an ergodic sequence \cite{Pa-Fi:92}.
Besides, there exists a non-negative and non-random measure $N_{H}$ on $%
\sigma (H),\;N(\mathbb{R})=1$ such that
\begin{equation}
\mathbf{E}\{f_{jj}(H)\}=\mathbf{E}\{f_{00}(H)\}=\int_{\sigma (H)}f(\lambda
)N_{H}(d\lambda ).  \label{IDS}
\end{equation}%
The measure $N_{H}$ is an important spectral characteristic of selfadjoint
ergodic operators known as the Integrated Density of States \cite{Pa-Fi:92}.
In particular, we have for any bounded $f:\sigma (H)\rightarrow \mathbb{R}$
with probability 1%
\begin{equation}
\lim_{|\Lambda |\rightarrow \infty \Lambda }|\Lambda |^{-1}\tr_\Lambda
f(H_{\Lambda })=\int_{\sigma (H)}f(\lambda )N_{H}(d\lambda ).  \label{lim}
\end{equation}%
This plays the role of the Law of Large Numbers for $\tr f(H_{\Lambda })$.

Accordingly, it is shown in \cite{Ki-Pa:15} (see also formula (\ref{SL})
below) that the leading term in an analog of (\ref{sf2}) for an ergodic
Schrodinger operator is always

\begin{equation}
|\Lambda |\int_{\sigma (H)}\varphi (a(\lambda ))N_{H}(d\lambda ).
\label{ltl}
\end{equation}%
On the other hand, the order of magnitude and the form of the subleading
term depend on the "amount of randomness" of an ergodic potential and on the
smoothness of $\varphi $ and, especially, $a$, see e.g. \cite%
{Bo-Si:90,De-Co:13,El-Co:17,Ki-Pa:15,Pa-Sl:18,So:13} for recent problems and
results.

In this paper we consider the discrete Schrodinger operator with random
i.i.d. potential, known also as the Anderson model. Thus, our quantity of
interest (\ref{trfl}) as well as the terms of its asymptotic form are random
variables in general (except the leading term (\ref{ltl}), which is not
random). Correspondingly, we will prove below two types of asymptotic trace
formulas, both having the subleading terms of the order $|\Lambda |^{1/2}$
(cf. (\ref{sf2})). The formulas of the first type are valid in the sense of
distributions, i.e., are analogs of the classical Central Limit Theorem (see
Theorems \ref{t:clt} and \ref{t:renyi}), while the formulas of the second
type are valid with probability 1, i.e., are analogs of the so called almost
sure Central Limit Theorem (see Theorem \ref{t:asclt}).

\begin{theorem}
\label{t:clt} Let $H$ be the ergodic Schrodinger operator (\ref{h}) -- (\ref%
{qq}) with a bounded i.i.d. potential and let $\sigma (H)$ be its spectrum.
Consider bounded functions $a:\sigma (H)\rightarrow \mathbb{R}$ and $%
\varphi:a(\sigma (H))\rightarrow \mathbb{R}$ and assume that $a$, $\varphi $
and $\gamma :=\varphi \circ a:\sigma (H)\rightarrow \mathbb{R}$ admit
extensions $\widetilde{a}$, $\widetilde{\varphi }$ and $\widetilde{\gamma }$
on the whole axis such that their Fourier transforms $\widehat{a}$, $%
\widehat{\varphi }$ and $\widehat{\gamma }$ satisfy the conditions
\begin{equation}  \label{afcon}
\int_{-\infty }^{\infty }(1+|t|^{\theta})|\widehat{f }(t)|dt<\infty,\;\;%
\theta >1, \;\; f=a,\varphi, \gamma.
\end{equation}
Denote%
\begin{equation}  \label{SL}
\Sigma _{\Lambda }=|\Lambda |^{-1/2}\Big(\tr_{\Lambda }\varphi (a_{\Lambda
}(H))-|\Lambda |\int_{\sigma (H)}\gamma (\lambda )N_H(d\lambda )\Big)
\end{equation}%
and
\begin{equation}
\sigma _{\Lambda }^{2}=\mathbf{E}\{\Sigma _{\Lambda }^{2}\}.  \label{sil}
\end{equation}%
Then:

(i) there exists the limit
\begin{equation}
\lim_{\Lambda \rightarrow \infty }\sigma _{\Lambda }^{2}=\sigma ^{2}
\label{sili}
\end{equation}%
where
\begin{equation}
\sigma ^{2}=\sum_{l\in \mathbb{Z}}C_{l}  \label{si1}
\end{equation}%
with
\begin{equation}
C_{j}=\mathbf{E}\{\overset{\circ }{\gamma }_{00}(H)\overset{\circ }{\gamma }%
_{jj}(H)\},\;\overset{\circ }{\gamma }_{jj}(H)=\gamma _{jj}(H)-\mathbf{E}%
\{\gamma _{jj}(H)\},  \label{gac}
\end{equation}%
and also%
\begin{eqnarray}
\sigma ^{2} &=&\mathbf{E}\left\{ \left( \mathbf{E}\left\{ A_{0}|\mathcal{F}%
_{0}^{\infty }\right\} -\mathbf{E}\left\{ A_{0}|\mathcal{F}_{1}^{\infty
}\right\} \right) ^{2}\right\}  \label{si2} \\
&=&\mathbf{E}\left\{ \mathbf{Var}\left\{ \mathbf{E\{}A_{0}|\mathcal{F}%
_{0}^{\infty }\right\} |\mathcal{F}_{1}^{\infty }\}\right\}  \notag
\end{eqnarray}%
where%
\begin{equation}
A_{0}=V_{0}\int_{0}^{1}\gamma _{00}^{\prime }(H|_{V_{0}\rightarrow uV_{0}})du
\label{aa0}
\end{equation}%
and $\mathcal{F}_{a}^{b},\;-\infty \leq a\leq b\leq \infty $ is the $\sigma $%
-algebra generated by $\{V_{j}\}_{j=a}^{b}$;

(ii) if $\gamma $ is non constant monotone function on the spectrum of $H,$
then
\begin{equation}
\sigma ^{2}>0  \label{sipos}
\end{equation}%
and we have
\begin{equation}
\mathbf{P\{}\sigma ^{-1}\Sigma _{\lbrack -M,M]}\in \Delta \}=\Phi (\Delta
)+o(1),\;M\rightarrow \infty ,  \label{dlim}
\end{equation}%
where $\Delta \subset \mathbb{R}$ is an interval and $\Phi $ is the standard
Gaussian law (of zero mean and unit variance).
\end{theorem}

\begin{remark}
The theorem is an extension of Theorem 2.1 of \cite{Ki-Pa:15}, where the
cases $a(\lambda )=\lambda $ and $\varphi (\lambda )=(\lambda -x_{0})^{-1}$
or $\varphi (\lambda )=\log (\lambda -x_{0}),\;x_{0}\notin \sigma (H)$ were
considered. In these cases $a,\varphi $ and $\gamma =\varphi \circ a$ are
real analytic on $\sigma (H)$ (see (\ref{spk})), hence admit real analytic
and fast decaying at infinity extensions to the whole line. Besides, $\gamma
$ is monotone on $\sigma (H)$, hence Theorem \ref{t:clt} applies.

It is worth also mentioning that conditions (\ref{afcon}) are not optimal in
general. Consider, for instance, the case where $\varphi (\lambda )= \chi
_{(-\infty,E]}(\lambda ),\;E\in \sigma (H)$, \; $a(\lambda )=\lambda $ with $%
\chi _{(-\infty,E]}$ being the indicator of $(-\infty,E] \subset \mathbb{R}$%
. Here $\gamma :=\varphi \circ a=$ $\chi _{(-\infty,E]}$ and%
\begin{equation}
\tr_{\Lambda }\varphi (a_{\Lambda }(H))=\tr_{\Lambda }\chi
_{(-\infty,E]}(H):=\mathcal{N}_{\Lambda }(E)  \label{CM}
\end{equation}%
is the number of eigenvalues of $H_{\Lambda }$ not exceeding $E$. It is
known that if the potential in $H$ is ergodic, then with probability 1%
\begin{equation*}
\lim_{|\Lambda |}|\Lambda |^{-1}\mathcal{N}_{\Lambda }(E)=N(E),
\end{equation*}%
where $N(E)$ is defined in (\ref{IDS}). This plays the role of the Law of
Large Numbers for $\mathcal{N}_{\Lambda }(E)$ \cite{Pa-Fi:92}. The Central Limit Theorem for $\mathcal{N}_{\Lambda }(E)$ is also known
\cite{Re:81}. Its proof is based on a careful analysis of a Markov chain
arising in the frameworks of the so called phase formalism, an efficient
tool of spectral analysis of the one dimensional Schrodinger operator \cite%
{Pa-Fi:92}. It can be shown that the theorem can also be proved following
the scheme of proof of Theorem \ref{t:clt}, despite that $\gamma $ is
discontinuous in this case. However, one has to use more sophisticated facts
on the Schrodinger operator with i.i.d. random potential, in particular the
bound
\begin{equation}\label{frmo}
\sup_{\varepsilon >0}\mathbf{E}\{|(H-E-i\varepsilon )_{jk}^{-1}|^{s}\}\leq
Ce^{-c|j-k|},
\end{equation}%
valid for some $s\in (0,1),\;C<\infty $ and $c>$ \cite{Ai-Wa:15} if the
probability law of potential possesses certain regularity, e.g. a bounded
density. The bound is one of the basic results of the spectral theory of the random Schrodinger operator,
implying the pure point character of the spectrum of $H$ and a number of its
other important properties. It is worth noting that the monotonicity of $%
\gamma $ on the spectrum remains true in this case. Thus, the monotonicity
of $\gamma$ seems a pertinent sufficient condition for the positivity of the
limiting variance.
\end{remark}

Here, however, is a version of the theorem, applicable to the case where $%
\gamma $ is a certain convex function on $\sigma (H)$.

\begin{theorem}
\label{t:renyi} Consider the functions $r_{\alpha }:[0,1]\rightarrow \lbrack
0,1]$ and $n_{F}:\mathbb{R}\rightarrow \lbrack 0,1]$ given by%
\begin{equation}
r_{\alpha }(\lambda )=(1-\alpha )^{-1}\log _{2}(\lambda ^{\alpha
}+(1-\lambda )^{\alpha }), \; \lambda \in [0,1], \;\alpha >0,  \label{ren}
\end{equation}%
and%
\begin{equation}
n_{F}(\lambda )=(e^{\beta (\lambda -E_{F})}+1)^{-1}, \; \lambda \in \mathbb{R%
},\;\beta >0,\;E_{F}\in \sigma (H)  \label{fer}
\end{equation}%
Assume that the \ random i.i.d. potential in (\ref{h}) -- (\ref{qq}) has
zero mean $\mathbf{E}\{V_{0}\}=0$ and that the support of its probability
law contains zero. Then the conclusions of Theorem \ref{t:clt} remain valid
for $\varphi =r_{\alpha }$ and $a=n_{F}$, i.e., the random variable $\Sigma
_{\Lambda }$ of (\ref{SL}) converges in distribution to the Gaussian random
variable of zero mean and a certain variance $\sigma ^{2}>0$.
\end{theorem}

\begin{remark}
\label{r:renyi} The quantity $\tr_{\Lambda }r_{\alpha }((n_{F}(H))_{\Lambda
})$ is known in quantum statistical mechanics and quantum information theory
as the R\'{e}nyi entanglement entropy of free fermions in the thermal state
of the inverse temperature $\beta ^{-1}>0$ and the Fermi energy $E_{F}$ and having $H$ as
the one body Hamiltonian, see, e,g. \cite{Ab-St:15,Ar-Co:14,Ei-Co:11}. An
important particular case where $\alpha =1$, hence $h_{1}(\lambda )=-\lambda
\log _{2}\lambda - (1-\lambda )\log _{2}(1-\lambda ), \;\lambda \in [0,1]$,
is known as the von Neumann entanglement entropy. One is interested in the
large-$|\Lambda |$ asympotic form of the entanglement entropy. In the
translation invariant case, i.e., for the case of constant potential in (\ref%
{h}) -- (\ref{qq}) one can use the Szeg\"o theorem (see (\ref{sf2}) and (\ref%
{szst1})) to find a two-term asymptotic formula for the entanglement
entropy. In this case the term proportional to $|\Lambda |$ in (\ref{sf2})
and (\ref{szst1}), i.e., to the one dimensional analog of the volume of the
spatial domain occupied by the system, is known as the volume law, while the
second term in (\ref{sf2}), which is independent of $|\Lambda |$, i.e.,
proportional to the one dimensional analog $\{-M,M\}$ of the surface area of
the domain, is known as the area law \cite{Ei-Co:11}. In view of the above
theorem we conclude that in the disorder case (random potential in $H$) the
leading term of the entanglement entropy is non-random and is again the
volume law while the subleading term is random, proportional to $|\Lambda
|^{1/2}$ and describes random fluctuations of the volume law. The $O(1)$ in $%
|\Lambda| $ term can also be found for some $\varphi $ and $a$ \cite%
{Ki-Pa:15}. It is random and is now the "subsubleading" term of the
asymptotic formula. Of particular interest is the zero-temperature case $%
\beta=\infty$, where $n_F=\chi_{-\infty,E}$ and this term is leading. We
refer the reader to recent works \cite%
{Ab-St:15,El-Co:17,Le-Co:13,Pa-Sl:14,Pa-Sl:18,Pi-So:18,So:13} for related
results and references.
\end{remark}


The above results can be viewed as stochastic analogs of the Szeg\"{o}
theorem (see more on the analogy in \cite{Ki-Pa:15} and below). It is
essentially a Central Limit Theorem in its traditional form, i.e., an
assertion on the convergence of distribution of an appropriately normalized
sums of random variables to the Gaussian random variable. In recent decades
there has been a considerable interest to the almost sure versions of
classical (distributional) limit theorems. The prototype of such theorems
dates back to P.Levy and P.Erdos and is as follows, see e.g. \cite{Be:98,De:13} for reviews.

Let $\{X_{l}\}_{l=1}^{\infty }$ be a sequence of i.i.d. random variables of
zero mean and unit variance. Denote $S_{m}=\sum_{l=1}^{m}$, $%
Z_{m}=m^{-1/2}S_{m}$. Then we have with probability 1
\begin{equation}
\frac{1}{\log M}\sum_{m=1}^{M}\frac{1}{m}\mathbf{1}_{\Delta }(Z_{m})=\Phi
(\Delta )+o(1), \; M \to \infty,  \label{asclt}
\end{equation}%
In other words, the random ("empirical") distribution of $Z_{m}$ converges
with probability 1 to the (non-random) Gaussian distribution.

On the other hand, the classical Central Limit Theorem implies%
\begin{equation}
\frac{1}{\log M}\sum_{m=1}^{M}\frac{1}{m}\mathbf{E}\{\mathbf{1}_{\Delta
}(Z_{m})\}=\Phi (\Delta )+o(1), \; M \to \infty,  \label{clt}
\end{equation}%
i.e., just the convergence of expectations of the random distributions on
the l.h.s. of (\ref{asclt}). Thus, replacing the expectation by the
logarithmic average, a sequence of random variables satisfying the CLT can
be observed along all its typical realizations.

The situation with the almost sure CLT (\ref{asclt}) for independent random
variables is rather well understood, see e.g. \cite{Be:98,De:13} and
references therein, while the case of dependent random variable is more
involved and diverse, see e.g. \cite{Ch-Go:07,Ib-Li:00,La-Ph:90,Pe-Sh:95}.
As in the case of classical CLT (\ref{clt}), the existing results concern
mostly the weakly dependent stationary sequences, e.g. strongly mixing
sequences. This and the approximation techniques developed \cite{Ib-Li:71}),
Section 18,3 allow us to prove an almost sure version of Theorem \ref{t:clt}.

\begin{theorem}
\label{t:asclt} We have with probability 1 under the conditions of Theorem %
\ref{t:clt}%
\begin{equation}  \label{gasclt}
\frac{1}{\log M}\sum_{m=1}^{M}\frac{1}{m}\mathbf{1}_{\Delta }(\sigma
^{-1}\Sigma _{\lbrack -m,m]})=\Phi (\Delta )+o(1), \; M \to \infty,
\end{equation}%
where, $\Sigma _{\lbrack -m,m]}$ is given by (\ref{SL}) with $\Lambda
=[-m,m] $, $\Delta \subset \mathbb{R}$ is an interval and $\Phi $ is the
standard Gaussian law.
\end{theorem}


\begin{remark}
Given a sequence $\{\xi _{m}\}_{m\geq 1}$ of random variables and a random
variable $\xi $, write%
\begin{equation}
\xi _{M}\overset{\mathcal{D}}{=}M^{1/2}\xi +o(M^{1/2}),\;M\rightarrow \infty
\label{cd1}
\end{equation}%
if we have%
\begin{equation}
\mathbf{P}\{\xi _{M}/M^{1/2}\in \Delta \}=G(\Delta ) + o(1),\;M\rightarrow
\infty,  \label{cd2}
\end{equation}%
where $G$ is the probability law of $\xi $, 
and write
\begin{equation}
\xi_{M}\overset{\mathcal{L}}{= }M^{1/2}\xi+o(M^{1/2}),\;M\rightarrow \infty
\label{cl1}
\end{equation}%
if we have with probability 1 (assuming that all $\{\xi_m, \;m\geq 1$ ) are
defined on the same probability space)%
\begin{equation}
\frac{1}{\log M}\sum_{m=1}^{M}I_\Delta(\xi_{m}/m^{1/2})=G(\Delta) +
o(1),\;M\rightarrow \infty.  \label{cl2}
\end{equation}%
Then, we can formulate Theorems \ref{t:clt} and (\ref{t:asclt}) in the form
similar to that of the Szeg\"{o} theorem (cf. (\ref{sf2})), namely as
\begin{align}
\tr_{\Lambda }\varphi (a_{\Lambda }(H))& \overset{\mathcal{D}}{=}|\Lambda
|\int_{\sigma (H)}\gamma (\lambda )N_{H}(d\lambda )  \label{szst1} \\
& +|\Lambda |^{1/2}\sigma ^{-1}\xi +o(|\Lambda |^{1/2}),\;|\Lambda
|=(2M+1)\rightarrow \infty  \notag
\end{align}%
for Theorems \ref{t:clt} and with probability 1 as%
\begin{align}
\tr_{\Lambda }\varphi (a_{\Lambda }(H))& \overset{\mathcal{L}}{=}|\Lambda
|\int_{\sigma (H)}\gamma (\lambda )N_{H}(d\lambda )  \label{szst2} \\
& +|\Lambda |^{1/2} \Phi(\sigma ^{-1}\Delta )+o(|\Lambda |^{1/2}),\;|\Lambda
|=(2M+1)\rightarrow \infty  \notag
\end{align}%
for Theorems \ref{t:clt}, i.e., as two-term "Szeg\"{o}-like" asymptotic
formulas valid in the sense of the $\mathcal{D}$- and the $\mathcal{L}$%
-convergence, the latter valid with probability 1. An apparent difference
between the Szeg\"{o} formula (\ref{sf2}) and its stochastic counterparts (%
\ref{szst1}) and (\ref{szst2}) is that the subleading term of the Szeg\"o
theorem is independent of $|\Lambda | $ while the subleading term of its
stochastic counterparts grows as $|\Lambda |^{1/2}$ although with stochastic
oscillations (see below).
\end{remark}

We will comment now on the errors bounds in the above asymptotic formulas.
We will mostly use known results on the rates of convergence for the both
CLT (\ref{cd2}) and (\ref{cl2}) with $\xi_m$ being the sum of i.i.d. random
variable (see (\ref{clt}) and (\ref{asclt})), despite that in our (spectral)
context the terms of the sum in (\ref{trfl}) are always dependent even if
the "output" potential is a collection of i.i.d. random variables. It seems
plausible that the error bounds for the i.i.d. case provide best possible
but not too overestimated versions of the error bounds for the case of
sufficiently weakly dependent terms. Known results on the sums of weakly
dependent random variables support this approach, see e.g. \cite%
{Be:98,De:13,Ch-Go:07,Ib-Li:00,La-Ph:90,Pe-Sh:95}.

Recall first that for the classical Szeg\"o (non-random) case (\ref{sf2}),
i.e., for the Toeplitz and convolution operators, the subleading term is $%
\Lambda $-independent and the error is just $o(1)$ in general. However, if $%
\varphi $ and $a$ are infinitely differentiable, one can construct the whole
asymptotic series in the powers of $|\Lambda |^{-1}$ \cite{Wi:85}.

On the other hand, it follows from the standard CLT for bounded i.i.d.
random variables (see (\ref{cd2})) and the Berry-Esseen bound that we have in
(\ref{cd2}) the error term $O(M^{-1/2})$ instead of $o(1)$,
and, hopefully, $O(|\Lambda^{-1/2}|)$ in the $\mathcal{D}$-convergence
stochastic analog (\ref{dlim}) of the Szeg\"o theorem.

As for the "point-wise" case treated in Theorem \ref{t:asclt}, we note first
that this is a "frequency"-type result, analogous to the Law of Large
Numbers or, more generally, to the ergodic theorem. This is clear from the
following observation on the well known Gaussian random processes \cite%
{Be:98}. Namely, let $W:[0,\infty )\rightarrow \mathbb{R}$ be the Wiener
process and $U:\mathbb{R}\rightarrow \mathbb{R}$ be the Uhlenbeck-Ornstein
process. They are related as $U(s)=e^{-s/2}W(e^{s}),\;s\in \mathbb{R}$, thus
\begin{equation*}
\frac{1}{\log M}\int_{1}^{M}\mathbf{1}_{\Delta }(W(t)/t^{1/2})dt=\frac{1}{%
\log M}\int_{0}^{\log M}\mathbf{1}_{\Delta }(U(s))ds.
\end{equation*}%
Since $U$ is ergodic and its one-point (invariant) distribution is the
standard Gaussian, the r.h.s. converges with probability 1 to $\Phi (\Delta
) $ as $T\rightarrow \infty $ according to the ergodic theorem. We obtained
the almost sure Central Limit Theorem for the Wiener process, the continuous
time analog of the sequence of i.i.d. Gaussian random variables, see (\ref%
{asclt}).

In view of this observation (explaining, in particular, the appearance of
the logarithmic average in the almost sure Central Limit Theorem) and the
Law of Iterated Logarithm we have to have with probability 1 in (\ref{cd2})
the oscillating error term $O((\log\log\log M/\log M)^{1/2})$ instead of $%
o(1)$, hence the error term $O((\log\log\log |\Lambda|/\log \Lambda)^{1/2})$
in the $\mathcal{L}$-convergence stochastic analog (\ref{gasclt}) of the
Szeg\"o theorem. More precisely, it follows from the invariance principle
that with probability 1 we have to have the additional terms $\widetilde{%
\sigma }W(\log M)+O(\log M^{1/2-\varepsilon }), \; M\rightarrow \infty$ in (%
\ref{cl2}) and, correspondingly, the terms
\begin{equation*}
\widetilde{\sigma }W(\log |\Lambda |)+O(|\log |\Lambda \}|^{1/2-\varepsilon
}),\;\log |\Lambda |\rightarrow \infty,
\end{equation*}
with $\widetilde{\sigma }>0$ and some $\varepsilon >0$ in (\ref{gasclt}).

\medskip
\medskip
We prove in this paper asymptotic formulas
for traces of certain random operators related to the restrictions
to the expanding intervals $\Lambda =[-M,M]\subset \mathbb{Z}%
,\;M\rightarrow \infty $ of the one dimensional discrete Schrodinger
operator $H$ assuming that its potential is a collection of random i.i.d.
variables. We do not use, however, a remarkable property of $H$, the pure
point character of its spectrum. This spectral type holds for any
bounded i.i.d. potential  \cite{Ai-Wa:15} and can be contrasted
with the absolute continuous type of the spectrum of $H$ with constant or
periodic potential. Moreover, if the common probability law of the on-site
potential is Lipschitzian, we have the bound (\ref{frmo}).
It can be shown that the use of the bound
makes the
conditions of our results somewhat weaker (it suffices to have $\theta =1$
in (\ref{afcon}), certain bounds  somewhat stronger ($O(1)$ instead $%
o(|\Lambda |^{1/2})$ in (\ref{dce}), $Ce^{-cp}$ instead $C/p^{\theta }$ in (%
\ref{fadec}), etc.) and proofs simpler (Lemmas \ref{l:ctu} and (\ref{l:fdec}%
) are not necessary). On the other hand, the bound (\ref{frmo}) holds only
under the condition of some regularity of the common probability law of
the i.i.d. potential (e.g., the Lipschitz continuity of its probability law). This is why
we prefer to use rather standard  spectral tools, somewhat less optimal
conditions (\ref{afcon}) on $a$ and $\varphi $ and somewhat more involved proofs but to have corresponding results valid for
a larger class of random i.i.d. potentials of
Theorems \ref{t:clt} and
\ref{t:asclt}.

It is worth noting, however, that the bound (\ref{frmo}) is an important necessary tool in the analysis of the large-$\Lambda$ behavior of $\mathrm{Tr}_{\Lambda} \varphi(a_\Lambda (H))$ with not too smooth $a$ and  $\varphi$, e.g. $a=n_F|_{\beta=\infty}=\chi_{[E_F,\infty)}$ with $n_F$ of (\ref{fer}) and $\varphi=r_\alpha, \; \alpha \leq 1$ with  $r_\alpha$ of (\ref{ren}) corresponding to the entanglement entropy of the ground state of free
disordered  fermions at zero temperature, see  \cite{El-Co:17,Pa-Sl:18} and references therein.


\section{Proof of Results}

\textbf{Proof of Theorem \ref{t:clt}} . It follows from (\ref{afcon}) and
Lemma \ref{l:fdif} that we have uniformly in potential
\begin{equation}
\tr_{\Lambda }\varphi (a_{\Lambda }(H))=\tr_{\Lambda }\varphi
(a(H))+o(|\Lambda |^{1/2}),\;|\Lambda |\rightarrow \infty .  \label{dce}
\end{equation}%
Hence, we obtain in view of (\ref{IDS}) and the definition (\ref{trl}) of $%
\tr_{\Lambda }$
\begin{align}
\tr_{\Lambda }\varphi (a_{\Lambda }(H)=|\Lambda |\int_{-\infty }^{\infty
}\gamma (\lambda )N(d\lambda )+\overset{\circ }{\gamma }_{\Lambda
}+o(|\Lambda |^{1/2}),  \notag
\end{align}%
where
\begin{align}
&\overset{\circ }{\gamma }_{\Lambda } :=\gamma _{\Lambda }-\mathbf{E}%
\{\gamma _{\Lambda }\},\;\; \gamma _{\Lambda }=\sum_{j\in \Lambda }(\gamma
_{jj}(H)-\mathbf{E}\{\gamma _{jj}(H)\}),  \label{GL} \\
&\mathbf{E}\{\gamma _{\Lambda }(H)\}) =|\Lambda |\int_{-\infty }^{\infty
}\gamma (\lambda )N_{H}(d\lambda )  \notag
\end{align}%
The above formulas reduce the proof of the theorem to that of the Central
Limit Theorem for $|\Lambda |^{-1/2}\gamma _{\Lambda }$, i.e., for the
sequence $\{\gamma _{jj}(H)\}_{j\in \mathbb{Z}}$. The sequence is ergodic
according to (\ref{eom}) for $j=k$. %
%

We  use in this case a general Central Limit Theorem for stationary
weakly dependent sequences given by Proposition \ref{p:ibli} with $%
X_{j}=\;V_{j},\;j\in \mathbb{Z}$ and $Y_{0}=\gamma _{00}(H)$. To verify the
approximation condition (\ref{dc}) of the proposition it is convenient to
write $V=(V_{<},V_{>})$, where $V_{<}=\{V_{j}\}_{|j|\leq p}\;$and $%
V_{>}=\{V_{j}\}_{|j|>p}$ are independent collections of independent random
variables whose probability laws we denote $P_{<}$ and $P_{>}$ so that the
probability law $P$ of $V$ is symbolically $P=P_{<}\cdot P_{>}$. Denoting $%
\gamma _{00}(H)=g(V_{<},V_{>})$, we have
\begin{align*}
& \mathbf{E}\{|\gamma _{00}(H)-\mathbf{E}\{\gamma _{00}(H)|\mathcal{F}%
_{-p}^{p}\}|\} \\
& =\int \Big|g(V_{<},V_{>})-\int g(V_{<},V_{>}^{\prime })P(dV_{>}^{\prime })%
\Big|P(dV_{>})P(dV_{<}) \\
& \leq \int \left( \int \Big|g(V_{<},V_{>})-g(V_{<},V_{>}^{\prime })\Big|%
P(dV_{>}^{\prime })\right) P(dV_{>})P(dV_{<}).
\end{align*}%
Applying to the difference in the third line of the above formula Lemma \ref%
{l:asin} with $f=\gamma_{00}$, we find that the expression in
the first line of the formula is bounded by $C/p^{\theta}$. Thus, the series (\ref%
{dc}) is convergent in our case.

This and Proposition (\ref{p:ibli}) imply the validity of (\ref{sili}) -- (%
\ref{gac}). The formula for the \ limiting variance (\ref{si2}) -- (\ref{aa0})  is proved in Lemma \ref{l:clt1}.


Let us prove now the positivity of the limiting variance $\sigma ^{2}$ (\ref%
{sipos}). According to (\ref{si2}) -- (\ref{aa0}), the hypothesis $\sigma
^{2}=0$ implies that for an almost every event from $\mathcal{F}_{1}^{\infty
}$ the expression%
\begin{equation}
V_{0}\int_{0}^{1}ds\ \mathbf{E\{}\gamma _{00}^{\prime }(H|_{V_{0}\rightarrow
uV_{0}})|\mathcal{F}_{1}^{\infty }\}du  \label{hyp}
\end{equation}%
is independent of $V_{0} \in \mathrm{supp} F$. Assume without loss of
generality that zero is in support of $F$. Then the above expression is
zero. On the other hand, if our i.i.d. random potential is non-trivial, then
there exists a non-zero point $V_{0}\neq 0$ in the support. If,
in addition, $\gamma^{\prime }$ does not change the sign on the spectrum of $%
H$ and is not zero, then (\ref{hyp}) cannot be zero, and we have a contradiction.

Now it suffices to use a general argument (see e.g. Theorem 18.6.1 of \cite%
{Ib-Li:71} or Proposition 3.2.9 of \cite{Pa-Sh:11}) to finish the proof of
Theorem \ref{t:clt}.$\blacksquare $


\bigskip \textbf{Proof of Theorem} \ref{t:renyi}. We will first use Theorem %
\ref{t:clt}. Indeed, according to (\ref{spk}) and (\ref{fer}), $a=n_{F}$ is
real analytic on the finite interval $K$ of (\ref{spk}) and admits a real analytic
and fast decaying at infinity extension to the whole axis. Besides, $%
a(K)=[a_{-},a_{+}],\;0<a_{-}<a_{+}<1$ is also finite, hence, $\varphi =r_{\alpha }$ of (\ref%
{ren}) is real analytic on $a(K)$ and admits a real analytic and fast
decaying at infinity extension to the whole axis. Thus, assertion (i) of
Theorem \ref{t:clt} is valid in this case.

We cannot, however, use assertion (ii) of Theorem \ref{t:clt}, since $\gamma
=r_{\alpha }\circ n_{F} $ is not monotone but convex on $K$. Here is another
argument proving the positivity (\ref{sipos}) of the limiting variance (\ref%
{si2}) -- (\ref{aa0}).

Assuming that the variance is zero and using the fact that zero is in
support of the probability law $F$ of the potential, we obtain from (\ref%
{si2}) -- (\ref{aa0}), as in the proof of Theorem \ref{t:clt}, that for
almost every event from $\mathcal{F}_{1}^{\infty }$ we have
\begin{equation*}
V_{0}\int_{0}^{1}\mathbf{E}\left\{\gamma _{00}^{\prime }(H_{0}(u))|\mathcal{F%
}_{1}^{\infty }\right\}du=0,\;V_{0}\in \mathrm{supp}F,
\end{equation*}%
where \ $H_{0}(u):=H|_{V_{0}\rightarrow uV_{0}}$. Integrating here by parts
with respect to $u$, we get%
\begin{align*}
&V_{0}\mathbf{E}\{\gamma _{00}^{\prime }(H_{0}(0)) |\mathcal{F}_{1}^{\infty
}\}+V_{0}\int_{0}^{1}
 \mathbf{E}\Big\{\frac{\partial }{\partial u}\gamma _{00}^{\prime
}(H_{0}(u))|\mathcal{F}_{1}^{\infty }\Big\}(1-u)du=0, \; \;V_{0}\in \mathrm{supp}F.
\end{align*}%
and since $\mathbf{E}\{V_{0}\}=0$ and $\gamma _{00}^{\prime }(H_{0}(0))$ is
independent of $V_{0}$, the expectation with respect to $V_{0}$ yields for
almost every event from $\mathcal{F}_{1}^{\infty }$

\begin{equation}
\int_{0}^{1}(1-u)du\int V_{0}\mathbf{E}\{\frac{\partial }{\partial u}\gamma
_{00}^{\prime }(H(u))|\mathcal{F}_{1}^{\infty }\}F(dV_{0})=0.  \label{intp}
\end{equation}%
We will now use the formula%
\begin{equation*}
\frac{\partial }{\partial u}\gamma _{00}^{\prime }(H(u)=V_{0}\int \int \frac{%
\gamma^{\prime } (\lambda _{1})-\gamma ^{\prime }(\lambda _{2})}{\lambda
_{1}-\lambda _{2}}\mu _{_{H(u)}}(d\lambda _{1})\mu _{_{H(u)}}(d\lambda _{2}),
\end{equation*}%
where $\mu _{_{H(u)}}(d\lambda )=(\mathcal{E}_{H(u)}(d\lambda ))_{00},\;$
and $\mathcal{E}_{H(u)}$ is the resolution of identity of $H(u)$.
Thus, $\mu _{0}\geq 0$ and $\mu _{0}(\mathbb{R})=1$. The formula can be
obtained by iterating twice the Duhamel formula (\ref{duh}).

Plugging the r.h.s. of the formula in (\ref{intp}) and recalling that $%
\gamma $ is strictly convex on the spectrum, hence $(\gamma (\lambda
_{1})-\gamma (\lambda _{2}))(\lambda _{1}-\lambda )^{-1}<0$, we conclude
that the r.h.s. of (\ref{intp}) is not zero. This implies the positivity of
the variance.$\blacksquare $

\bigskip \textbf{Proof of Theorem \ref{t:asclt}} . As in the proof of Theorem %
\ref{t:clt} we will start with passing from $\tr_{\Lambda }\varphi
(a_{\Lambda }(H))$ to $\tr_{\Lambda } \varphi (a(H))=\tr_{\Lambda } \gamma
(H)$ with the error $o(|\Lambda |^{1/2})$ by using (\ref{afcon}) and Lemma %
\ref{l:fdif} (see (\ref{dce})), thereby reducing the proof of the theorem to
the proof of the almost sure CLT for $|\Lambda |^{-1/2}\gamma _{\Lambda }$
(see \ref{GL}) i.e., for the same ergodic sequence $\{\gamma
_{jj}(H)\}_{j\in \mathbb{Z}}$ as in Theorem \ref{t:clt}.

Our further proof is essentially based on that in \cite{Pe-Sh:95} of the
almost sure CLT for ergodic strongly mixing sequences (see(\ref{alk})) and
on the procedure of approximation of general ergodic sequences by strongly
mixing sequences (see \ref{dc})) given in \cite{Ib-Li:71}, Section 18.3.
In particular, according to Proposition \ref{p:pesh1} (see Theorem 1 in \cite%
{Pe-Sh:95}), it suffices to prove the bound
\begin{equation}
\mathbf{Var}\left\{ \frac{1}{\log M}\sum_{m=1}^{M}\frac{1}{m}f\left(
Z_{m}\right) \right\} =O(1/(\log M)^{\varepsilon }),\;M\rightarrow \infty \;
\label{varf}
\end{equation}%
for any bounded Lipschitzian $f$ \ (see (\ref{lip})),
\begin{equation}
Z_{m}=\mu _{m}^{-1/2}\Sigma _{[ -m,m]},\;\mu _{m}=2m+1  \label{zm}
\end{equation}%
and some $\varepsilon >0$.

To this end we denote
\begin{equation}  \label{Y}
\overset{\circ }{\gamma }_{jj}(H):=\gamma _{jj}(H)-\mathbf{E}\{\gamma
_{jj}(H)\}=Y_{j},\;j\in \mathbb{Z}
\end{equation}
and introduce for every positive integer $s$ the ergodic sequences $\{\xi
_{j}^{(s)}\}_{j\in \mathbb{Z}}$ and $\{\eta _{j}^{(s)}\}_{j\in \mathbb{Z}}$
with%
\begin{equation}
\xi _{j}^{(s)}=\mathbf{E}\{Y_{j}|\mathcal{F}_{j-s}^{j+s}\},\;\;\eta
_{j}^{(s)}=Y_{j}-\xi _{j}^{(s)}.  \label{xeg}
\end{equation}%
Denote also%
\begin{equation*}
F_{M}=\frac{1}{\log M}\sum_{m=1}^{M}\frac{1}{m}f\left( Z_{m}\right)
\end{equation*}%
and%
\begin{equation}
F_{M}^{(s)}=\frac{1}{\log M}\sum_{m=1}^{M}\frac{1}{m}f\left(
Z_{m}^{(s)}\right) ,\;Z_{m}^{(s)}=Z_{m}|_{Y_{j}\rightarrow \xi _{j}^{(s)}}.
\label{fms}
\end{equation}%
We have then from the elementary inequality $\mathbf{Var}\{\xi\}\leq 2\mathbf{%
Var}\{\eta\}+2\mathbf{Var}\{\xi-\eta\}$ and (\ref{lip}):%
\begin{eqnarray}
\mathbf{Var}\{F_{M}\} &\leq &2\mathbf{Var}\{F_{M}^{(s)}\}+2\mathbf{Var}%
\{F_{M}-F_{M}^{(s)}\}  \label{var1} \\
&\leq &2\mathbf{Var}\{F_{M}^{(s)}\}+\frac{2C_1 ^{2}}{\log M}\sum_{m=1}^{M}%
\frac{1}{m}\mathbf{Var}\{R_{m}^{(s)}\},  \notag
\end{eqnarray}%
where $C_1$ is defined in (\ref{lip}) and%
\begin{equation}
R_{m}^{(s)}:=Z_{m}-Z_{m}^{(s)}=\mu_m^{-1/2}\sum_{|j|\leq m}\eta
_{j}^{(s)}. \label{zms}
\end{equation}%
Recall now that an ergodic sequence is said to be strongly mixing if%
\begin{equation}
\alpha _{k}:=\sup_{A\in \mathcal{F}_{-\infty }^{n},\;B\in \mathcal{F}%
_{k+n}^{\infty }}|P(AB)-P(A)P(B)|\rightarrow 0  \label{alk}
\end{equation}%
as $k\rightarrow \infty $ through positive values and $\alpha _{k}$ is
called the mixing coefficient.

Since the random potential is a sequence of i.i.d. random variables, the
sequence $\{\xi _{j}^{(s)}\}_{j\in \mathbb{Z}}$ of (\ref{xeg}) is strongly mixing and its
mixing coefficient is (see (\ref{alk}))
\begin{equation}
\alpha _{k}^{(s)}=\left\{
\begin{array}{cc}
\leq 1, & k\leq 2s \\
0, & k>2s.%
\end{array}%
\right.  \label{axi}
\end{equation}%
We are going to bound the first term on the right of (\ref{var1}) by using
Lemma 1 of \cite{Pe-Sh:95} on the almost sure CLT for strongly mixing
sequences and we will deal with the second term on the right of (\ref{var1})
by using the sufficiently good approximation of $\{Y_{j}\}_{j\in \mathbb{Z}}$
of (\ref{Y}) by $\{\xi _{j}^{(s)}\}_{j\in \mathbb{Z}}$ of (\ref{xeg}) as $%
s\rightarrow \infty $ following from Lemma \ref{l:asin}. Note that similar
argument has been already used in the proof of Theorem \ref{t:clt}, see (\ref%
{dc}) in Proposition \ref{p:ibli} and Theorem 18.6.3 in \cite{Ib-Li:71}.
This is obtained in Lemmas \ref{l:sims} and \ref{l:vxi} below for $%
M\rightarrow \infty $ and $s\rightarrow \infty $.
They allow us to continue (\ref{var1}) as%
\begin{equation*}
\mathbf{Var}\{F_{M}\}=O(\log s/\log M)+O(1/s^{\theta -1}),
\end{equation*}%
where $\theta >1$ (see (\ref{afcon})). Choosing here $s=(\log
M)^{1-\varepsilon },\;\varepsilon \in (0,1)$, we obtain (\ref{varf}), hence,
the theorem.$\blacksquare $


\section{Auxiliary Results}

We start with a general Central Limit Theorem for ergodic
sequences of random variables, see \cite{Ib-Li:71}, Theorems 18.6.1 -
18.6.3, more precisely. with its version involving i.i.d. random variables.

\begin{proposition}
\label{p:ibli} Let $\{X_{j}\}_{j\in \mathbb{Z}}$ be i.i.d. random variables,
$\mathcal{F}_{a}^{b}$ be the $\sigma $-algebra generated by $%
\{X_{j}\}_{j=a}^{b}$, $Y_{0}$ be a function measurable with respect to $%
\mathcal{F=F}_{-\infty }^{\infty }$. Denote $T$ the standard shift
automorphism ($X_{j+1}(\omega)=X_j(T\omega)$) and set $Y_{j}(\omega
)=Y_{0}(T^{j}\omega )$. Assume that

(i) \ $Y_{0}$ is bounded;

(ii)%
\begin{equation}
\sum_{p=1}^{\infty }\mathbf{E}\{|Y_{0}-\mathbf{E}\{Y_{0}|\mathcal{F}%
_{-p}^{p}\}|\}<\infty .  \label{dc}
\end{equation}%
Then

(a) \ $\sigma ^{2}:=\sum_{k=0}^{\infty }\mathbf{Cov}\{Y_{0},Y_{k}\}<\infty ;$

(b) if $\sigma ^{2}>0$, then%
\begin{equation}
|\Lambda |^{-1/2}\sum_{|j|\leq M}Y_{j}  \label{ns}
\end{equation}%
converges in distribution to the Gaussian random variable of zero mean and
variance $\sigma ^{2}$.
\end{proposition}

The proof of the proposition is based on the proof of the CLT for strongly
mixing ergodic sequences (see (\ref{alk})) and on the approximation of more
general ergodic sequences by strongly mixing sequences provided by condition
(\ref{dc}).

We will also need several facts on the one-dimensional discrete Schrodinger
operator with bounded potential.

We recall first the Duhamel formula for the difference of two one-parametric
groups $U_{1}(t)=e^{itA_{1}}$ and $U_{1}(t)=e^{itA_{1}}$ corresponding to
two bounded operators $A_{1}$ and $A_{2}:$%
\begin{equation}  \label{duh}
U_{2}(t)-U_{1}(t)=i\int_{0}^{|t|}U_{2}(t-s)(A_{2}-A_{1})U_{1}(s)ds,\;t\in
\mathbb{R}.
\end{equation}

\begin{lemma}
\label{l:ctu} Let $H=H_{0}+V$ be the one-dimensional discrete Schrodinger
operator with real-valued bounded potential, $U(t)=e^{itH}$ be the
corresponding unitary group and $\{U_{jk}(t)\}_{j,k\in \mathbb{Z}}$ be the
matrix of $U(t)$. Then we have for any $t\in \mathbb{R}$ and $\delta >0$%
\begin{equation}
|U_{jk}(t)|\leq e^{-\delta |j-k|+s(\delta )|t|},\;s(\delta )=2\sinh \delta .
\label{ctu}
\end{equation}
\end{lemma}

\begin{proof}
Introduce the diagonal operator $D=\{D_{jk}\}_{j,k\in \mathbb{Z}}$, with $%
D_{jk}=e^{\rho j}\delta _{jk}$, $\rho \in \mathbb{R}$ and consider%
\begin{equation*}
DU(t)D^{-1}=e^{itDHD^{-1}} =e^{itH+itQ},
\end{equation*}%
where%
\begin{equation*}
Q=DHD^{-1}-H=DH_{0}D^{-1}-H_{0}.
\end{equation*}%
Since $H_{0}$ is the operator of second finite difference with the symbol $%
-2\cos p,\;p$ $\in \mathbb{T}$, the symbol of $Q$ is%
\begin{equation*}
-2\cos (p+i\rho )+2\cos p=-2\cos p(\cosh \rho -1)+2i\sin p\sinh \rho .
\end{equation*}%
Hence $Q=Q_{1}+iQ_{2}$, where $Q_{1}$ and $Q_{2}$ are selfadjoint operators
and%
\begin{equation*}
||Q_{2}||\leq 2\sinh |\rho |.
\end{equation*}%
Now, denoting $A_{2}=H+Q_{1}+iQ_{2}$ and $A_{1}=H+Q_{1}$, iterating the
Duhamel formula (\ref{duh}) and using $||e^{itA_{1}}||=1$, we obtain%
\begin{equation*}
||e^{itH_{1}-tQ_{2}}||\leq e^{|t|\;||Q_{2}||}=e^{|t|2\sinh |\rho |}.
\end{equation*}%
This and the relation%
\begin{equation*}
(DU(t)D^{-1})_{jk}=e^{\rho j}U_{jk}(t)e^{-\rho k}
\end{equation*}%
imply (\ref{ctu}).
\end{proof}

\begin{remark}
\label{r:ctr} Bound (\ref{ctu}) is an analog of the Combes-Thomas bound for
the resolvent $\{((H-z)^{-1})_{jk}\}_{j,k \in \mathbb{Z}}$ of $H$ and the
above proof uses an essentially same argument as that in the proof of
the bound, see e.g. \cite{Ai-Wa:15}.
\end{remark}

\begin{lemma}
\label{l:fdec} Let $H=H_{0}+V$ be the one-dimensional discrete Schrodinger
operator with real-valued potential and $a:\mathbb{R}\rightarrow \mathbb{R}$
admits the Fourier transform $\widehat{a}$ and%
\begin{equation}
\int_{-\infty }^{\infty }(1+|t|^{\theta })|\widehat{a}(t)|dt<\infty
,\;\theta >0..\;  \label{fta}
\end{equation}%
If $A=a(H)=\{F_{jk}\}_{j,k\in \mathbb{Z}}$, then we have%
\begin{equation}
|A_{jk}|\leq C/|j-k|^{\theta },\; C< \infty, \;j \neq k.  \label{fdec}
\end{equation}
\end{lemma}

\begin{proof}
It follows from the spectral theorem
\begin{equation}  \label{frep}
A=a(H)=\int_{-\infty }^{\infty }\widehat{a}(t)U(t)dt,
\end{equation}%
hence, we have for any $T>0$
\begin{align*}
& A_{jk}=\int_{-\infty }^{\infty }\widehat{a}(t)U_{jk}(t)dt \\
& =\int_{|t|\leq T}\widehat{a}(t)U(t)dt+\int_{|t|\geq T}\widehat{a}%
(t)U(t)dt=I_{1}+I_{2}.
\end{align*}%
We have further
\begin{equation*}
|I_{1}|\leq e^{-\delta |j-k|+s(\delta )T}\int_{-\infty }^{\infty } |\widehat{%
a }(t)|dt
\end{equation*}%
by using Lemma \ref{l:ctu} and
\begin{equation*}
|I_{2}|\leq \frac{1}{T^{\theta}}\int_{-\infty }^{\infty }(1+|t|^{\theta})|%
\widehat{a}(t)|dt
\end{equation*}%
by condition (\ref{fta}) of the lemma.

Now, choosing $T=\frac{\delta}{s}|j-k| -\theta \log |j-k|$, we obtain (\ref%
{fdec}).
\end{proof}

\begin{lemma}
\label{l:fdif} Let $A=\{A_{jk}\}_{j,k\in \mathbb{Z}}$ be bounded selfadjoint
operator in $l^{2}(\mathbb{Z})$ such that
\begin{equation}
|A_{jk}|\leq C/|j-k|^{\theta},\;C<\infty ,\;\theta>1  \label{adec}
\end{equation}%
and $A_{\Lambda }=\chi _{\Lambda }A\chi _{\Lambda }=\{A_{jk}\}_{j,k\in
\Lambda }$ be its restriction to $\Lambda $. Then for any $f:\mathbb{R}%
\rightarrow \mathbb{C}$ admitting the Fourier transform $\widehat{f}$ such
that%
\begin{equation}
\int_{-\infty }^{\infty }(1+|t|)|\widehat{f}(t)|dt<\infty  \label{ft1}
\end{equation}%
we have uniformly in $V$ satisfying (\ref{qq})
\begin{equation}
\big|\tr\chi _{\Lambda }f(A_{\Lambda })\chi _{\Lambda }-\tr\chi _{\Lambda
}f(A)\chi _{\Lambda }\big|=o(L^{1/2}),\;L\rightarrow \infty .  \label{tfdif}
\end{equation}
\end{lemma}

\begin{proof}
Consider $A_{\Lambda }\oplus A_{\overline{\Lambda }}, \; \overline{\Lambda }=%
\mathbb{Z}\setminus \Lambda $ and

\begin{equation*}
A-A_{\Lambda }\oplus A_{\overline{\Lambda }}=\left(
\begin{array}{cc}
0 & \chi _{\Lambda }A\chi _{\overline{\Lambda }} \\
\chi _{\overline{\Lambda }}A\chi _{\Lambda } & 0%
\end{array}%
\right) .
\end{equation*}%
Thus, writing an analog of (\ref{frep}) for $A$ instead of $H$
and using the Duhamel formula (\ref{duh}), we obtain%
\begin{align}
&f(A)-f(A_{\Lambda }\oplus A_{\mathbb{Z}\setminus \Lambda })  \label{fdif} 
=\int_{-\infty }^{\infty }\widehat{f}(t)dt\int_{0}^{|t|}U(t-s)(\chi _{%
\overline{\Lambda }}A\chi _{\Lambda }+\chi _{\Lambda }A\chi _{\overline{%
\Lambda }})U_{\Lambda }(s)\oplus U_{\overline{\Lambda }}(s)ds,  
\end{align}
and
\begin{align}\label{trdif}
&\mathrm{Tr}\ \chi _{\Lambda }f(A)\chi _{\Lambda }-\mathrm{Tr}\ \chi
_{\Lambda }f(A_{\Lambda })\chi _{\Lambda }   
=\int_{-\infty }^{\infty }\widehat{f}(t)dt\int_{0}^{t}\mathrm{Tr}\chi
_{\Lambda }U_{\Lambda }(s)\chi _{\Lambda }U(t-s) \chi _{\overline{\Lambda}
}A\chi _{\Lambda }ds.  
\end{align}%
Denoting
\begin{align}
B=\chi _{\Lambda }U_{\Lambda }(s)\chi _{\Lambda }U(t-s):l^{2}(\mathbb{Z}%
)\rightarrow l^{2}(\Lambda )  \label{B}
\end{align}%
we can write the integrand $J$ in (\ref{trdif}) as
\begin{equation}
J=\sum_{j\in \Lambda ,k\in \overline{\Lambda }}A_{kj}B_{jk},  \label{AB}
\end{equation}%
hence 
\begin{equation*}
|J|\leq \sum_{j\in \Lambda }\Big(\sum_{k\in \overline{\Lambda }%
}|A_{kj}|^{2}\sum_{k\in \overline{\Lambda }}|B_{jk}|^{2}\Big)^{1/2}\leq
\sum_{j\in \Lambda }\Big(\sum_{k\in \overline{\Lambda}}|A_{kj}|^{2}\sum_{k%
\in \mathbb{Z}}|B_{jk}|^{2}\Big)^{1/2}.
\end{equation*}
We have in view of (\ref{B})
\begin{align*}
&\hspace{-1cm}\sum_{k%
\in \mathbb{Z}}|B_{jk}|^{2}=(BB^{\ast })_{jj}
\\&=(U_{\Lambda }(s)\chi _{\Lambda }U(t-s)U^{\ast }(t-s)\chi
_{\Lambda }U^{\ast }(s))_{jj}
=(\chi _{\Lambda }U_{\Lambda }(s)U_{\Lambda
}^{\ast }(s))_{jj}=1
\end{align*}%
since $U(t-s)$ is unitary in $l^{2}(\mathbb{Z})$ and $U_{\Lambda }(s)$ is
unitary in $l^{2}(\Lambda )$. Thus, we have in view of (\ref{adec})
\begin{equation*}
|J|\leq \sum_{j\in \Lambda }\Big(\sum_{k\in \overline{\Lambda }}|A_{kj}|^{2}%
\Big)^{1/2}\leq C\sum_{j\in \Lambda }\Big(\sum_{k\in \overline{\Lambda }%
}|k-j|^{-2\theta }|\Big)^{1/2}=o(L^{1/2})
\end{equation*}%
and (\ref{fdif}) follows. Note that for $\theta >3/2$ the r.h.s. of the
above bound is $O(1)$.
\end{proof}

\medskip Similar result was obtained in \cite{La-Sa:96} by another method.

\begin{lemma}
\label{l:asin} Let $H_{1}$ and $H_{2}$ be the one dimensional discrete
Schrodinger operators with bounded potentials $V_{1}$ and $V_{2}$
coinciding within the integer valued interval $[-p,p]$. Consider $f:\mathbb{R}%
\rightarrow \mathbb{C}$ whose Fourier transform $\widehat{f}$ is such that
\begin{equation}
\int_{-\infty }^{\infty }(1+|t|^{\theta })|\widehat{f}(t)|dt<\infty
,\;\theta >1.  \label{fdec1}
\end{equation}%
Then we have
\begin{equation}
|f_{00}(H_{1})-f_{00}(H_{2})|\leq C/p^{\theta },  \label{fadec}
\end{equation}%
where $C$ is independent  of $V_{1}$ and $V_{2}$.
\end{lemma}

\begin{proof}
We denote
\begin{equation*}
V_{1}=\{V^{\prime }\}_{|j|>s}\cup \{V_{j}\}_{|j|\leq
s},\;V_{2}=\{V_{j}^{\prime \prime }\}_{|j|>s}\cup \{V_{j}\}_{|j|\leq s},
\end{equation*}%
$U^{(1)}(t)=e^{itH_{1}}$ and $U^{(2)}(t)=e^{itH_{2}}$ and use (\ref{frep})
and the spectral theorem to write for any $T>0$%
\begin{eqnarray}  \label{df}
|f_{00}(H_{2})-f_{00}(H_{1})| &\leq &\int_{|t|\leq T}|\widehat{f}%
(t)||U_{00}^{(2)}(t)-U_{00}^{(1)}(t)|dt  \label{def} \\
&&+\int_{|t|\geq T}|\widehat{f}%
(t)||U_{00}^{(2)}(t)-U_{00}^{(1)}(t)|=:I_{1}+I_{2}.  \notag
\end{eqnarray}%
We have then by the Duhamel formula (\ref{duh}) and (\ref{qq})
\begin{align*}
I_{1}& \leq \int_{|t|\leq T}|\widehat{f}(t)|dt\int_{0}^{|t|}%
\sum_{|j|>s}|U_{0j}^{(2)}(t-t^{\prime })(V_{j}^{\prime \prime
}-V_{j}^{\prime })U_{j0}^{(1)}(t^{\prime })|dt^{\prime }. \\
& \leq 2\overline{V}\int_{|t|\leq T}|\widehat{f}(t)|dt\int_{0}^{|t|}%
\sum_{|j|>s}|U_{0j}^{(2)}(t-t^{\prime })||U_{j0}^{(1)}(t^{\prime
})|dt^{\prime }.
\end{align*}%
We will use now Lemma \ref{l:ctu} implying%
\begin{equation}
I_{1}\leq 2\overline{V}e^{-2\delta p+s(\delta )T}\int_{-\infty }^{\infty
}(1+|t|)\widehat{f}(t)|dt.  \label{I1}
\end{equation}%
To estimate $I_{2}$ of (\ref{df}), we write
\begin{equation}
I_{2}\leq 2\int_{|t|\geq T}|\widehat{f}(t)|dt\leq \frac{2}{T^{\theta }}%
\int_{-\infty }^{\infty }(1+|t|^{\theta })|\widehat{f}(t)|dt.  \label{I2}
\end{equation}%
Choosing now in (\ref{I1}) and (\ref{I2}) $T=2\delta p/s(\delta )-\theta
\log p$, we obtain (\ref{fadec}).
\end{proof}

\begin{lemma}
\label{l:clt1} Consider a bounded $\gamma :\mathbb{R}\rightarrow \mathbb{R}$
admitting the Fourier transform $\widehat{\gamma }$ such that
\begin{equation}
\int_{-\infty }^{\infty }(1+|t|)|\widehat{\gamma }(t)|dt<\infty .
\label{gcond}
\end{equation}%
and set $\gamma (H)=\{\gamma _{jk}(H)\}_{j,k\in \mathbb{Z}}$, where $H$ is
the one-dimensional Schrodinger operator (\ref{h}) -- (\ref{qq}) with random
i.i.d. potential. 
%
Let $\gamma _{\Lambda }$ be defined in (\ref{GL}) and
\begin{equation}  \label{sila}
\sigma _{\Lambda }^{2}=|\Lambda|^{-1}\mathbf{Var}\{\gamma _{\Lambda }\}
\end{equation}%
Then there exists the limit%
\begin{equation}  \label{siM}
\sigma ^{2}:=\lim_{\Lambda \rightarrow \infty }\sigma _{\Lambda }^{2} =%
\mathbf{E}\{(\mathcal{M}^{(0)})^{2}\}
\end{equation}%
where%
\begin{align}\label{M}
\mathcal{M}^{(0)}=&\mathbf{E}\{A_{0}(V_{0},\{V_{j}\}_{j\neq 0})|\mathcal{F}%
_{0}^{\infty }\}\\&-\int_{-\infty }^{\infty }\mathbf{E\{}A_{0}(V'_0,\{V_{j}\}_{j%
\neq 0})F(dV_0')|\mathcal{F}_{0}^{\infty }\}
\notag
\end{align}%
and%
\begin{equation}
A_{0}(V_{0},\{V_{j}\}_{j\neq 0})=V_{0}\int_{0}^{1}(\gamma ^{\prime
}(H)|_{V_{0}\rightarrow uV_{0}})_{00}du  \label{A0}
\end{equation}
\end{lemma}

\begin{proof}
It is convenient to consider%
\begin{equation}
\tau _{\Lambda }:=\mathrm{Tr}_{\Lambda }\gamma (H_{\Lambda })  \label{tal}
\end{equation}%
instead of $\gamma _{\Lambda }$ of (\ref{GL}). It follows from Lemma \ref%
{l:fdif} that
\begin{equation}
\sigma^2_\Lambda =|\Lambda |^{-1}\mathbf{Var}\{\tau _{\Lambda }\} +
o(1),\;|\Lambda |\rightarrow \infty .  \label{dtg}
\end{equation}%
To deal with $\mathbf{Var}\{\tau _{\Lambda }\}$ we will use a simple version
of the martingale techniques (see e.g. \cite{Pa-Sh:11}, Proposition 18.1.1),
according to which if $\{X_{j}\}_{j=-M}^{M}$ are the i.i.d. random
variables, $\Phi :\mathbb{R}^{2M+1}\rightarrow \mathbb{R}$ is bounded and $%
\Phi =\Phi (X_{-M},X_{-M+1},...,X_{M})$, then%
\begin{align}
& \mathbf{Var}\{\Phi \}:=\{|\Phi -\mathbf{E}\{\Phi \}|^{2}\}  \label{mart1}
=\sum_{|m|\leq M}\mathbf{E}\{|\Phi ^{(m)}-\Phi
^{(m+1)}|^{2}\},  
\end{align}%
where
\begin{equation}
\Phi ^{(m)}=\mathbf{E}\{\Phi |\mathcal{F}_{m}^{M}\},\;\Phi ^{(-M)}=\Phi
,\;\Phi ^{(M+1)}=\mathbf{E}\{\Phi \}.  \label{mart2}
\end{equation}

We choose in (\ref{mart1}) -- (\ref{mart2}) $X_{j}=V_{j},\;|j|\leq M$ and $%
\Phi =\tau _{\Lambda }$ \ (see (\ref{tal})) and we obtain%
\begin{align}
|\Lambda |^{-1}\mathbf{Var}\{\tau _{\Lambda }\} &=|\Lambda
|^{-1}\sum_{m=-M}^{M}\mathbf{E}\{|\mathcal{M}_{\Lambda }^{(m)}|^{2}\},
\label{varm} \\
\mathcal{M}_{\Lambda }^{(m)}& =\tau _{\Lambda }^{(m)}-\tau _{\Lambda
}^{(m+1)}  \notag
\end{align}%
where (see (\ref{mart2}))
\begin{equation}
\tau _{\Lambda }^{(m)}=\mathbf{E}\{\tau _{\Lambda }|\mathcal{F}%
_{m}^{M}\},\;\tau _{\Lambda }^{(-M)}=\tau _{\Lambda },\;\tau _{\Lambda
}^{(M+1)}=\mathbf{E}\{\tau _{\Lambda }\}.  \label{tlL}
\end{equation}%
By using the formula%
\begin{eqnarray}
\tau _{\Lambda }-\tau _{\Lambda }|_{V_{m}=0} &=&\int_{0}^{1}du\frac{\partial
}{\partial u}\mathrm{Tr_{\Lambda }}\ \gamma (H_{\Lambda }|_{V_{m}\rightarrow
uV_{m}})  \label{tvt0} \\
&=&V_{m}\int_{0}^{1}du(\gamma ^{\prime }(H_{\Lambda }|_{V_{m}\rightarrow
uV_{m}}))_{mm}=:A_{\Lambda }(V_{m},\{V\}_{j\neq m}).  \notag
\end{eqnarray}%
we can write%
\begin{align}  \label{Mmla}
\mathcal{M}_{\Lambda }^{(m)} &=\mathbf{E}\{A_{\Lambda }(V_{m},\{V\}_{j\neq
m})|\mathcal{F}_{m}^{M}\} \\
&-\int_{-\infty }^{\infty }A_{\Lambda }(V_{m}',\{V\}_{j\neq m})|\mathcal{F}%
_{m}^{M}\}F(dV_m').  \notag
\end{align}%
Let us show now that
\begin{equation}
\lim_{\Lambda \rightarrow \infty }|\Lambda |^{-1}\sum_{m\in \Lambda }%
\mathbf{E}\{|\mathcal{M}_{\Lambda }^{(m)}|^2\}=\mathbf{E}\{|\mathcal{M}^{(0)}|^2\},
\label{amml}
\end{equation}%
where for any $m\in \mathbb{Z}$
\begin{align}  \label{lim1}
\mathcal{M}^{(m)}&=\mathbf{E}\{A(V_{m},\{V\}_{j\neq m})|\mathcal{F}%
_{m}^{\infty}\} \\
&-\int_{-\infty }^{\infty } \mathbf{E}\{A(V_{m}',\{V\}_{j\neq m})|\mathcal{F}%
_{m}^{\infty}\}F(dV_m').  \notag
\end{align}%
Note first that since $A_{\Lambda }(V_{m},\{V\}_{j\neq m})$ does not depend
on $\{V_j\}_{|j| > M }$, we can replace $\mathcal{F}_m^M$
by $\mathcal{F}_m^\infty$
Next, it is easy to see that $\mathcal{M}_{\Lambda }^{(m)}$
is bounded in $\Lambda $ and $V$, thus the proof of (\ref{amml}) reduces to
the proof of validity with probability 1 of the relation%
\begin{equation}
\lim_{\Lambda \rightarrow \infty ,\;\mathrm{dist}(m,\{M,-M\})\rightarrow
\infty}\mathcal{M}_{\Lambda }^{(m)}=\mathcal{M}^{(m)}.  \label{lim2}
\end{equation}%
Note that $\mathcal{M}^{(m)}$ of (\ref{lim1}) differs from its prelimit form
$\mathcal{M}^{(m)}_\Lambda$ of (\ref{Mmla}) by the replacement of $%
H_{\Lambda}$ by $H$ in the r.h.s. of (\ref{tvt0}).

Indeed, if (\ref{lim2}) is valid, then we can replace $\mathcal{M}_{\Lambda
}^{(m)}$ by $\mathcal{M}^{(m)}$ in the l.h.s. of (\ref{amml}) and then take
into account that $V$ is a collection of i.i.d. random variables, hence \ $%
\mathbf{E}\{|\mathcal{M}^{(m)}|^2\}=$ $\mathbf{E}\{|\mathcal{M}^{(0)}|^2\}$ for any $%
m\in \mathbb{Z}$.

To prove (\ref{lim2}) we will use a version of formula (\ref{fdif}) with $%
f=\gamma ^{\prime }$ implying for $m\in \Lambda $%
\begin{equation}
(\gamma ^{\prime }(H)-\gamma ^{\prime }(H_{\Lambda }))_{mm}=i\int_{-\infty
}^{\infty }t\widehat{\gamma }(t)dt\int_{0}^{|t|}(U_{\Lambda }(t-s)\chi
_{\Lambda }H\chi _{\overline{\Lambda }}U(s))_{mm}ds.  \label{gamd}
\end{equation}%
Taking into account that the non-zero entries of $\chi _{\Lambda }H\chi _{\overline{\Lambda }}$ are
$-\delta _{j,M}\delta _{k,M+1}$ and $-\delta
_{j,-M}\delta _{k,-(M+1)}$, $|j|\leq M,|k|>M$, we obtain%
\begin{equation}
|(U_{\Lambda }(t-s)\chi _{\Lambda }H\chi _{\overline{\Lambda }%
}U(s))_{mm}|\leq |U_{M+1,m}(s)|+|U_{-(M+1),m}(s)|.  \label{www}
\end{equation}%
We write now the integral over $t$ in (\ref{gamd}) as the sum of the
integral $I_{1}$ over $|t|\leq T$ and that $I_{2}$ over $|t|\geq T$ for some
$T$, cf. the proofs of Lemmas \ref{l:fdec} and \ref{l:asin}. We have by
Lemma \ref{l:ctu} and (\ref{www})%
\begin{align*}
|I_{1}|\leq & 2e^{-\delta d+sT}\int_{|t|
\leq T}|t|^{2}|\widehat{\gamma }(t)|dt
\\  \leq & 2e^{-\delta d+sT}T\int_{|t|
\leq T}|t||\widehat{\gamma }(t)|dt,\;
d=\mathrm{dist}(m,\{M,-M\})
\end{align*}%
and by (\ref{www}) and the unitarity of $U(s)$%
\begin{equation*}
|I_{2}|\leq 2\int_{|t|\geq T}|t||\widehat{\gamma }(t)|dt.
\end{equation*}%
Now, choosing $sT=\delta d/2$ and taking into account (\ref{gcond}), we
obtain (\ref{lim2}), hence, the assertion of the lemma.
\end{proof}

\begin{proposition}
\label{p:pesh1} Let $\{X_{j}\}_{j\in \mathbb{Z}}$ be a sequence of random
variables on the same probability space with $\mathbf{E}\{X_{l}\}=0,\;%
\mathbf{E}\{X_{l}^{2}\}<\infty $. Put (cf. (\ref{SL} -- (\ref{sil}))%
\begin{equation}
S_{m}=\sum_{|j|\leq m}X_{l},\;Z_{m}=\mu _{m}^{-1/2}S_{m}, \;\mu
_{m}=2m+1,\;\sigma _{m}^{2}=\mathbf{E}\{Z_{m}^{2}\}  \label{ssz}
\end{equation}%
and assume:

(i) $Z_{m}\overset{\mathcal{D}}{\rightarrow }\xi_\sigma ,\;m\rightarrow \infty $,
where $\xi $ is the Gaussian random variable of zero mean and variance $%
\sigma ^{2}>0$;

(ii) for every bounded Lipschitz $f$:%
\begin{equation}
|f(x)|\leq C,\;|f(x)-f(y)|\leq C_1 |x-y|  \label{lip}
\end{equation}%
there exists $\varepsilon >0$, such that%
\begin{equation*}
\mathbf{Var}\left\{ \frac{1}{\log M}\sum_{m=1}^{M}\frac{1}{m}f\left(
Z_{m}\right) \right\} =O(1/(\log M)^{\varepsilon }),\;M\rightarrow \infty .
\end{equation*}%
Then $\{X_{j}\}_{j\in \mathbb{Z}}$ satisfies the almost sure Central Limit
Theorem, i.e., we have with probability 1%
\begin{equation*}
\lim_{M\rightarrow \infty }\frac{1}{\log M}\sum_{m=1}^{M}\frac{1}{m} \mathbf{%
1}_{\Delta }(Z_{m})=\Phi(\sigma^{-1}\Delta ),
\end{equation*}%
where $\Delta $ is an interval and $\Phi$ is the standard Gaussian law.
\end{proposition}

The proposition is a version of Theorem 1 of \cite{Pe-Sh:95} where the case
of semi-infinite stationary sequences $\{X_{l}\}_{l=1}^{\infty }$ was
considered. For another criterion of the validity of the almost sure CLT see \cite{Ib-Li:00}).

\begin{lemma}
\label{l:sims} Let $\{\xi _{j}^{(s)}\}_{j\in \mathbb{Z}}$ be defined in \ (%
\ref{xeg}), $Z _{m}^{(s)}$ be defined in (\ref{zm}) and (\ref{fms}) and
\begin{equation*}
(\sigma _{m}^{(s)})^{2}=\mathbf{E}\{(Z_{m}^{(s)})^{2}\}.
\end{equation*}%
Then we have:

(i) for every $m=1,2,...$
\begin{equation*}
|\sigma _{m}^{(s)}-\sigma _{m}|\leq C/s^{(\theta -1)},
\end{equation*}%
where $\sigma_m>0$ is given in (\ref{ssz}), $C$ is independent of $m$ and $s$ and $\theta >1$ is given in
(\ref{afcon});

(ii) for any $\delta >0$ there exists $m_{0}>0$ and $s_{0}>0\,$
such that
\begin{equation*}
|\sigma _{m}^{(s)}-\sigma |\leq \sigma \delta
\end{equation*}%
if $m>m_0$ and $s>s_0$ and $\sigma >0$ is given
in Theorem \ref{t:clt};;

(iii) 
for every $m=1,2,...$
\begin{equation*}
\mathbf{E}\{((R^{(s)})_m)^2\} \leq C/s^{\theta-1},
\end{equation*}%
where $C$ is independent of $m$ and $s$
and $\theta >1$ is given in (\ref{afcon}).
\end{lemma}

\begin{proof}
The lemma is a version of the obvious fact $\lim_{s\rightarrow
\infty }\xi
_{m}^{(s)}=Y_{m}$ valid with probability 1 for every $m$ and following from (%
\ref{xeg}).

(i).  Since $\{Y_{j}\}_{j\in \mathbb{Z}}$ and $\{\xi _{j}^{(s)}\}_{j\in \mathbb{Z}%
} $ are ergodic sequences, we can write%
\begin{eqnarray}
\sigma _{m}^{2}-(\sigma _{m}^{(s)})^{2} &=&\sum_{|l|\leq 2s}(1-|l|/\mu
_{m})(C_{l}-C_{l}^{(s)})  \label{dsig} 
+\sum_{2s<|l|\leq 2m}(1-|l|/\mu _{m})C_{l}, 
\end{eqnarray}%
where $C_{l}=\mathbf{E}\{Y_{0}Y_{l}\}$ and $C_{l}^{(s)}=\mathbf{E}\{\xi
_{0}^{(s)}\xi _{l}^{(s)}\}$ are the correlation functions of the
corresponding sequences (see (\ref{gac}) and we took into account (\ref{axi}%
) implying that $C_{l}^{(s)}=0,\;|l|>2s$ (and that the second term
on the right is present only if $m>2s$). Since $Y_{j}=\xi
_{l}^{(s)}+\eta _{l}^{(s)}
$, we have%
\begin{equation*}
C_{l}-C_{l}^{(s)}=\mathbf{E}\{\xi _{0}^{(s)}\eta _{l}^{(s)}\}+\mathbf{E}%
\{\xi _{j}^{(s)}\eta _{0}^{(s)}\}+\mathbf{E}\{\eta _{0}^{(s)}\eta
_{l}^{(s)}\}.
\end{equation*}%
Since $\gamma$ is bounded, it follows from (\ref{Y}) --
(\ref{xeg}) that
\begin{equation}\label{coes}
|\mathbf{E}\{\xi _{0}^{(s)}\eta _{l}^{(s)}\}| \leq C \psi_s, \;
|\mathbf{E}\{\xi _{j}^{(s)}\eta _{0}^{(s)}\}\leq C \psi_s, \;
|\mathbf{E}\{\eta _{0}^{(s)}\eta _{l}^{(s)}\}| \leq C \psi_s,
\end{equation}
where
\begin{equation}
\psi _{s}=\mathbf{E}\{|\eta _{0}^{(s)}|\}  \label{psi}
\end{equation}%
and by (\ref{xeg}) and Lemma \ref{l:asin}%
\begin{equation}
\psi _{s}=O(1/s^{\theta }),\;\theta >1.  \label{ces}
\end{equation}%
This and (\ref{dsig}) imply uniformly in $m$%
\begin{align}
& |\sigma _{m}^{2}-(\sigma _{m}^{(s)})^{2}|\leq \sum_{|l|\leq
2s}|C_{l}-C_{l}^{(s)}|+\sum_{2s<|l|\leq 2m}|C_{l}|  \label{sim2} \\
& \hspace{1cm}=O(s\psi _{s})+O(1/s^{(\theta -1)})=O(1/s^{(\theta
-1)}),\;s\rightarrow \infty .  \notag
\end{align}%

(ii).  $\sigma _{m}$ of (\ref{ssz}) is strictly positive for every
$m$ and
according to Theorem \ref{t:clt} and Lemma \ref{l:clt1}%
\begin{equation*}
\lim_{m\rightarrow \infty }\sigma _{m}^{2}=\sigma ^{2}>0.
\end{equation*}%
This and (\ref{sim2}) imply the assertion.

(iii). The ergodicity of $\{\eta^{(s)}_{j}\}_{j\in \mathbb{Z}}$ implies
(cf. (\ref{dsig}))
\begin{equation}\label{varms}
\mathbf{Var}\{R_m^{(s)}\} = \sum_{|l| \leq 2m} (1-|l|/\mu
_{m})\mathbf{E}\{\eta _{0}^{(s)}\eta _{l}^{(s)}\} =\sum_{|l| \leq
6s} + \sum_{6s<|l| \le 2m }.
\end{equation}
It follows then from the proof of Proposition \ref{p:ibli} (see \cite%
{Ib-Li:71}, Theorem 18.6.3) and (\ref{axi}
) that
\begin{equation}
\mathbf{E}\{\eta _{0}^{(s)}\eta _{l}^{(s)}\} \leq C\psi
_{[l/3]}),\;|l|>6s,  \label{ees}
\end{equation}%
We will use now (\ref{coes}) -- (\ref{ces}) in the first sum on
the r.h.s. of (\ref{varms}) and (\ref{ces}) and (\ref{ees}) in the
second sum (cf. (\ref{sim2})) to get the bound
\begin{equation*}
\mathbf{Var}\{R_m^{(s)}\} \leq C(s\psi_s + \sum_{|l| > 6s}
\psi_{[l/3]}) \leq C/s^{\theta-1}.
\end{equation*}
proving the assertion.
\end{proof}
\begin{lemma}
\label{l:vxi} Let $\{\xi _{j}^{(s)}\}_{j\in \mathbb{Z}}$ and $F_{M}^{(s)}$
be defined in \ (\ref{xeg}) and (\ref{fms}) respectively. Then we have:
\begin{equation*}
\mathbf{Var}\left\{ F_{m}^{(s)}\right\} =O(\log s/\log M),\;s\rightarrow
\infty ,\;M\rightarrow \infty .\;
\end{equation*}
\end{lemma}

\begin{proof}
Repeating almost literally the proof of Lemma 1 in \cite{Pe-Sh:95} (where
the case of semi-infinite strongly mixing sequences $\{X_{l}\}_{l=1}^{\infty
}$ was considered), we obtain
\begin{align}
& \hspace{-1.5cm}\mathbf{Var}\left\{ F_{m}^{(s)}\right\} \leq \frac{%
C^{\prime }}{\log M^{2}}+\frac{C^{\prime \prime }}{\log M}\sum_{m=1}^{.M}%
\frac{\alpha _{m}^{(s)}}{m}  \label{pesh} \\
& +\frac{C^{\prime \prime \prime }}{(\log M)^{2}}\sum_{m=1}^{M-1}
(2m^{-1/2}\sigma _{2m}^{(s)}+\mathbf{E}\{m^{-1}|\xi
_{0}^{(s)}|\})\sum_{l=m+1}^{M}\frac{1}{l^{3/2}}  \notag
\end{align}%
where $C^{\prime },C^{\prime \prime },C^{\prime \prime \prime }$ depend only
on $C$ in (\ref{lip}) and $\alpha _{m}^{(s)}$ is the mixing coefficient (\ref%
{alk}) of $\{\xi _{l}^{(s)}\}_{l\in \mathbb{Z}}$ given by (\ref{axi}). In
view of (\ref{axi}) the second term is bounded by
\begin{equation}
\frac{C^{\prime \prime }}{\log M}\sum_{m=1}^{2s}\frac{1}{m}\leq
C_{2}^{^{\prime }}\frac{\log 2s}{\log M}=O(\log s/\log M)  \label{ster}
\end{equation}%
as $M\rightarrow \infty $ and $s\rightarrow \infty $.

Consider now the third term of the r.h.s. of (\ref{pesh}). It follows from (%
\ref{xeg}) and our assumption on the boundedness of $Y_{0}$ that the
contribution of $\mathbf{E}\{|\xi _{0}^{(s)}|\}$ is $O(1/(\log M)^{2})$.
Next, given an $M$-independent $M_{0}$ of Lemma \ref{l:sims} (ii), we write
\begin{equation*}
\sum_{m=1}^{M-1}2\sigma _{2m}^{(s)}\sum_{l=m+1}^{M}\frac{1}{l^{3/2}}%
=\sum_{m=1}^{M_{0}}\sum_{l=m+1}^{M}+\sum_{m=M_{0}+1}^{M-1}\sum_{l=m+1}^{M}.
\end{equation*}%
The first double sum on the right is bounded in $M$ in view of Lemma \ref%
{l:sims} (i) and the fact that $\sigma _{m},\;m=1,2,...,M_{0}$ are bounded
(e.g. $\sigma _{m}\leq \mathbf{E}^{1/2}\{Y_{0}^{2}\}$). The second double
sum is in view of Lemma \ref{l:sims} (ii)
\begin{equation*}
O\left( \sum_{m=M_{0}+1}^{M-1}\frac{1}{m^{1/2}}\sum_{l=m+1}^{M}\frac{1}{%
l^{3/2}}\right) =O\left( \log M\right) .
\end{equation*}%
Hence, the third term on the right of (\ref{pesh}) is $O\left( 1/\log
M\right) $. This and (\ref{ster}) imply that the r.h.s. of (\ref{pesh}) is $%
O(\log s/(\log M))$.
\end{proof}



\end{document}